\documentclass[12pt]{article}%
\usepackage{amsfonts}
\usepackage{amsmath}
\usepackage{amssymb}
\usepackage{graphicx}%
\providecommand{\U}[1]{\protect\rule{.1in}{.1in}}
\newtheorem{theorem}{Theorem}

\newtheorem{lemma}[theorem]{Lemma}

\newtheorem{proposition}[theorem]{Proposition}
\newtheorem{remark}[theorem]{Remark}

\newenvironment{proof}[1][Proof]{\noindent\textbf{#1.} }{\ \rule{0.5em}{0.5em}}
\begin{document}

\title{Construction of Multivariate Gaussian Weyl--Heisenberg Frames, (I)}
\author{Maurice de Gosson\thanks{Financed by the Austrian Research Agency FWF
(Projektnummer P20442-N13).}\\\textit{Universit\"{a}t Wien, NuHAG}\\\textit{Fakult\"{a}t f\"{u}r Mathematik }\\\textit{A-1090 Wien}}
\maketitle

\begin{abstract}
(This Note replaces a former note which contains an incorrect proof) .Let
$\phi$ be an arbitrary generalized Gaussian (squeezed coherent state),
$\Lambda_{\alpha\beta}=(\alpha_{1}\mathbb{Z}\times\cdot\cdot\cdot\times
\alpha_{n}\mathbb{Z)\times}(\beta_{1}\mathbb{Z}\times\cdot\cdot\cdot
\times\beta_{n}\mathbb{Z)}$ a rectangular lattice. We show that there exists a
positive definite symplectic matrix $M$ (depending on $\phi$) such that the
multivariate Weyl--Heisenberg system $\mathcal{G}(\phi,M\Lambda_{\alpha\beta
})$ is a frame. In a forthcoming Note we will prove a converse to this result.

\end{abstract}

\section{Introduction}

Let $\phi\in L^{2}(\mathbb{R}^{n})$, $\phi\neq0$ and a lattice $\Lambda
\subset\mathbb{R}^{2n}$. For $z_{0}\in\mathbb{R}^{2n}$ define the Heisenberg
operator $\widehat{T}(z_{0}):L^{2}(\mathbb{R}^{n})\longrightarrow
L^{2}(\mathbb{R}^{n})$ by
\begin{equation}
\widehat{T}(z_{0})\psi=e^{\tfrac{i}{\hslash}(p_{0}\cdot x-\tfrac{1}{2}%
p_{0}\cdot x_{0})}\psi(x-x_{0}) \label{heiwe}%
\end{equation}
($\hbar$ is a positive constant, usually taken to be $1/2\pi$ in
time-frequency analysis, and to $h/2\pi$ in quantum mechanics; $h$ is Planck's
constant). The set $\mathcal{G}(\phi,\Lambda)=\{\widehat{T}(z)\phi:z\in
\Lambda\}$ is called a Weyl--Heisenberg (or Gabor) system. If $\mathcal{G}%
(\phi,\Lambda)$ is a frame in $L^{2}(\mathbb{R}^{n})$, i.e. \cite{Gro} if
there exist $a,b>0$ such that
\begin{equation}
a||\psi||^{2}\leq\sum_{z\in\Lambda}|(\psi|\widehat{T}(z)\phi)|^{2}\leq
b||\psi||^{2} \label{frame1}%
\end{equation}
for every $\psi\in L^{2}(\mathbb{R}^{n})$ then $\mathcal{G}(\phi,\Lambda)$ is
called a Weyl--Heisenberg (or Gabor) frame.

\begin{remark}
The function $z\longmapsto(\psi|\widehat{T}(z)\phi)$ is, up to the factor
$(2\pi\hbar)^{n}$, the cross-ambiguity function $A(\psi,\phi)=F_{\sigma}%
W(\psi,\phi)$ ($F_{\sigma}$ the symplectic Fourier transform and $W(\psi
,\phi)$ the cross-Wigner distribution).
\end{remark}

A particularly interesting situation occurs when one chooses a Gaussian window
$\phi$ because Gaussians play a privileged role in both time-frequency
analysis and quantum mechanics \cite{Folland,Birk,Birkbis,Gro}. A classical
result is the following necessary and sufficient condition in the case $n=1$,
due to Lyubarski \cite{Lyu} and Seip and Wallst\'{e}n \cite{sewa92}:

\begin{proposition}
\label{Prop1}Let $\phi_{1}(x)=(\pi\hbar)^{-1/4}e^{-x^{2}/2\hbar}$ (the
\textquotedblleft fiducial coherent state\textquotedblright) with
$x\in\mathbb{R}$ and $\Lambda_{\alpha\beta}=\alpha\mathbb{Z}\times
\beta\mathbb{Z}$. The Gabor system $\mathcal{G}(\phi_{1},\Lambda_{\alpha\beta
})$ is a frame for $L^{2}(\mathbb{R})$ if and only if $\alpha\beta<2\pi\hbar$.
\end{proposition}

This result has the following non-trivial extension, proven in \cite{AB,AB2}:

\begin{proposition}
\label{Prop2}Let $\phi=\phi_{1}\otimes\cdot\cdot\cdot\otimes\phi_{1}$ and
$\Lambda_{\alpha\beta}=(\alpha_{1}\mathbb{Z}\times\cdot\cdot\cdot\times
\alpha_{n}\mathbb{Z)\times}(\beta_{1}\mathbb{Z}\times\cdot\cdot\cdot
\times\beta_{n}\mathbb{Z)}$. Then $\mathcal{G}(\phi,\Lambda_{\alpha\beta})$ is
a frame if and only if $\alpha_{j}\beta_{j}<2\pi\hbar$ for $1\leq j\leq n$.
\end{proposition}

The problem of constructing multivariate Weyl--Heisenberg systems
$\mathcal{G}(\phi,\Lambda)$ with an arbitrary Gaussian $\phi$ and lattice
$\Lambda$ is reputedly difficult and has been tackled by many authors (see the
comments in \cite{GroLyu}, and the review in \cite{Gro2}). That problem
however becomes more easily tractable if one recasts it in terms of
phase-space objects such that Heisenberg operators and cross-Wigner function,
which allows one to uses the full power of the symplectic covariance machinery
familiar to mathematical physicists working in phase space quantum mechanics
\cite{Birk,Birkbis}. We are going to show that:

\begin{proposition}
\label{Theorem}Let the lattice $\Lambda_{\alpha\beta}=\alpha\mathbb{Z}%
^{n}\times\beta\mathbb{Z}^{n}$ be defined as above; let $\phi_{X,Y}$ be a
\textquotedblleft squeezed coherent state\textquotedblright, that is a
Gaussian of the type
\begin{equation}
\phi_{X,Y}(x)=\left(  \tfrac{1}{\pi\hbar}\right)  ^{n/4}(\det X)^{1/4}%
e^{-\tfrac{1}{2\hbar}(X+iY)x^{2}} \label{Gauss1}%
\end{equation}
where $X+iY$ is a complex symmetric $n\times n$ matrix with real part $X>0$.
Let $G=S^{T}S$ be the positive definite symplectic matrix where
\begin{equation}
S=%
\begin{pmatrix}
X^{1/2} & 0\\
X^{-1/2}Y & X^{-1/2}%
\end{pmatrix}
. \label{bi}%
\end{equation}
The Weyl--Heisenberg system $\mathcal{G}(\phi_{X,Y},G^{-1/2}\Lambda
_{\alpha\beta})$ is a frame if and only if $\alpha_{j}\beta_{j}<2\pi\hbar$ for
$1\leq j\leq n$.
\end{proposition}

\section{Two Lemmas}

Let $\operatorname*{Mp}(2n,\mathbb{R})$ be the metaplectic group; we denote by
$\pi^{\operatorname*{Mp}}:\operatorname*{Mp}(2n,\mathbb{R})\longrightarrow
\operatorname*{Sp}(2n,\mathbb{R})$ the natural projection onto the symplectic
group. Recall that $\operatorname*{Mp}(2n,\mathbb{R})$ is a double cover of
$\operatorname*{Sp}(2n,\mathbb{R})$ consisting of unitary operators on
$L^{2}(\mathbb{R}^{n})$ \cite{Folland,Birk,Birkbis}. We recall the following
covariance property of the Heisenberg operator:%
\begin{equation}
\widehat{S}\widehat{T}(z)=\widehat{T}(Sz)\widehat{S}\text{ \ , \ }%
S=\pi^{\operatorname*{Mp}}(\widehat{S})\text{.} \label{cov1}%
\end{equation}

\begin{lemma}
\label{Lemma1}Let $\mathcal{G}(\phi,\Lambda)$ be a Weyl--Heisenberg system,
and $\widehat{S}\in\operatorname*{Mp}(2n,\mathbb{R})$, $S=\pi
^{\operatorname*{Mp}}(\widehat{S})$. Then $\mathcal{G}(\phi,\Lambda)$ is a
frame in $L^{2}(\mathbb{R}^{n})$ if and only if $\mathcal{G}(\widehat{S}%
\phi,S\Lambda)$ is a frame in $L^{2}(\mathbb{R}^{n})$.
\end{lemma}

\begin{proof}
(See \cite{Birkbis}, Chapter 8). We have, using (\ref{cov1}),%
\[
\sum_{z\in S\Lambda}|(\psi|\widehat{T}(z)\widehat{S}\phi)|^{2}=\sum_{z\in
S\Lambda}|(\psi|\widehat{S}\widehat{T}(S^{-1}z)\phi)|^{2}=\sum_{z\in\Lambda
}|(\widehat{S}^{-1}\psi|\widehat{T}(z)\phi)|^{2}%
\]
hence the result since $||\widehat{S}^{-1}\psi||=||\psi||$.
\end{proof}

The second Lemma gives an explicit formula for the Wigner transform of a
squeezed coherent state. Recall that for $\phi\in L^{2}(\mathbb{R}^{n})$
\begin{equation}
W\psi(z)=\left(  \tfrac{1}{2\pi\hbar}\right)  ^{n}\int_{\mathbb{R}^{n}%
}e^{-\frac{i}{\hbar}p\cdot y}\psi(x+\tfrac{1}{2}y)\overline{\psi(x-\tfrac
{1}{2}y)}dy. \label{wigo1}%
\end{equation}

\begin{lemma}
\label{Lemma2}Let $\phi_{X,Y}$ be the Gaussian (\ref{Gauss1}). We have%
\begin{equation}
W\phi_{X,Y}(z)=\left(  \tfrac{1}{\pi\hbar}\right)  ^{n}e^{-\tfrac{1}{\hbar
}Gz^{2}} \label{phagauss}%
\end{equation}
where $G\in\operatorname*{Sp}(2n,\mathbb{R})$ is the positive-definite matrix
\begin{equation}
G=%
\begin{pmatrix}
X+YX^{-1}Y & YX^{-1}\\
X^{-1}Y & X^{-1}%
\end{pmatrix}
=S^{T}S \label{G}%
\end{equation}
where the symplectic matrix $S$ is given by (\ref{bi}).
\end{lemma}

\begin{proof}
See \cite{Birk,Birkbis}.
\end{proof}

\section{Proof of Proposition \ref{Theorem}}

Recall \cite{Birk,Birkbis} the following symplectic covariance property of the
Wigner transform:
\begin{equation}
W(\psi,\phi)(S^{-1}z)=W(\widehat{S}\psi,\widehat{S}\phi)(z). \label{cov2}%
\end{equation}
We note that $\mathcal{G}(\phi,\Lambda)$ is a frame if and only if
$\mathcal{G}(c\phi,\Lambda)$ is a frame when $c\in\mathbb{C}$ is a complex number.

The matrix $G$ defined by (\ref{G}) in Lemma \ref{Lemma2} is both
positive-definite and symplectic hence there exists $U\in\operatorname*{Sp}%
(2n,\mathbb{R})\cap O(2n,\mathbb{R})$ such that $UGU^{T}=D$ is diagonal
\cite{Folland,Birk}. Let $\lambda_{1},...,\lambda_{2n}$ be the eigenvalues of
$G$. Since the eigenvalues of a positive-definite symplectic matrix occur in
pairs $(\lambda_{j},1/\lambda_{j})$ we may assume that $\lambda_{1}\geq
...\geq\lambda_{n}\geq1$, hence these eigenvalues $\lambda_{j}$ can be ordered
as follows: $\lambda_{1}\geq\cdot\cdot\cdot\geq\lambda_{n}\geq1\geq\lambda
_{n}^{-1}\geq\cdot\cdot\cdot\geq\lambda_{1}^{-1}$. We thus have%
\[
G=U^{T}DU=U^{T}%
\begin{pmatrix}
\Delta & 0\\
0 & \Delta^{-1}%
\end{pmatrix}
U
\]
with
\begin{equation}
\Delta=\operatorname{diag}(\lambda_{1},...,\lambda_{n}) \label{delta}%
\end{equation}
and hence%
\[
G^{-1/2}=U^{T}%
\begin{pmatrix}
\Delta^{-1/2} & 0\\
0 & \Delta^{1/2}%
\end{pmatrix}
U\in\operatorname*{Sp}(2n,\mathbb{R}).
\]

We are going to show that the Gaussian $\phi_{X,Y}$ becomes a tensor product
of elementary one-dimensional Gaussians if transformed by a suitable
metaplectic operator. Let in fact $\widehat{S_{G}}\in\operatorname*{Mp}%
(2n,\mathbb{R})$ be one of the two metaplectic operators such that
$\pi^{\operatorname*{Mp}}(\widehat{S_{G}})=G^{-1/2}$. Formulas (\ref{phagauss}%
) and (\ref{cov2}) imply that
\begin{equation}
W(\widehat{S_{G}}\phi_{X,Y})(z)=W(\phi_{X,Y})(G^{-1/2}z)=\left(  \tfrac{1}%
{\pi\hbar}\right)  ^{n}e^{-\tfrac{1}{\hbar}|z|^{2}} \label{wsphi}%
\end{equation}
and hence $\widehat{S_{G}}\phi_{X,Y}=c\phi_{1}\otimes\cdot\cdot\cdot
\otimes\phi_{1}$ for some complex constant with $|c|=1$. Thus, in view of
Lemma \ref{Lemma1}, $\mathcal{G}(\phi_{X,Y},G^{-1/2}\Lambda_{\alpha\beta})$ is
a frame if and only if $\mathcal{G}(\phi_{1}\otimes\cdot\cdot\cdot\otimes
\phi_{1},\Lambda_{\alpha\beta})$ is a frame. But this is the case if and only
if $\alpha_{j}\beta_{j}<2\pi\hbar$ for $1\leq j\leq n$ in view of Proposition
\ref{Prop2}.

E-Mail: maurice.de.gosson@univie.ac.at


\begin{thebibliography}{99}                                                                                               %


\bibitem {AB}A. Bourouihiya, Beurling Weighted Spaces, Product-Convolution
Operators, and the Tensor Product of Frames, Doctoral Dissertation, University
of Maryland, College Park USA, directed by Professor John J. Benedetto (2006).

\bibitem {AB2}A. Bourouihiya, The tensor product of frames, Sampl. Theory
Signal Image Process. \textbf{7}(1) (2008), 65--76.

\bibitem {Folland}G. B. Folland,\textit{ Harmonic Analysis in Phase space},
Annals of Mathematics studies, Princeton University Press, Princeton, N.J. (1989)

\bibitem {Birk}M. de Gosson, \textit{Symplectic Geometry and Quantum
Mechanics}, Birkh\"{a}user, Basel, series \textquotedblleft Operator Theory:
Advances and Applications\textquotedblright\ (subseries: \textquotedblleft
Advances in Partial Differential Equations\textquotedblright), Vol.
\textbf{166} (2006).

\bibitem {Birkbis}M. de Gosson, Symplectic Methods in Harmonic Analysis;
Applications to Mathematical Physics, Birkh\"{a}user, 2011 (in press).

\bibitem {physreps}M. de Gosson and F. Luef, \textit{Symplectic Capacities and
the Geometry of Uncertainty: the Irruption of Symplectic Topology in Classical
and Quantum Mechanics}, Physics Reports, \textbf{484} (2009), 131--179 179DOI 10.1016/j.physrep.2009.08.001.

\bibitem {Gro}K. Gr\"{o}chenig, \textit{Foundations of Time-Frequency
Analysis}, Birkh\"{a}user, Boston, (2000).

\bibitem {Gro2}K. Gr\"{o}chenig, Multivariate Gabor Frames and Sampling of
Entire Functions of Several Variables, Preprint (November 2010).

\bibitem {GroLyu}K. Gr\"{o}chenig and Yu. Lyubarskii, \textit{Gabor
(super)frames with Hermite functions}, Math. Annal. \textbf{345} (2009), 267--286.

\bibitem {Gromov}M. Gromov, \textit{Pseudoholomorphic curves in symplectic
manifolds}, Invent. Math., \textbf{82} (1985), 307--347.

\bibitem {Lyu}Yu. I. Lyubarskii, \textit{Frames in the Bargmann space of
entire functions}, In Entire and subharmonic functions, Amer. Math. Soc.,
Providence RI\ (1992), 167--180.

\bibitem {Polter}L. Polterovich, \textit{The Geometry of the Group of
Symplectic Diffeomorphisms}, Lectures in Mathematics, Birkh\"{a}user, (2001).

\bibitem {sewa92}K. Seip and R. Wallst\'{e}n, \textit{Density theorems for
sampling and interpolation in the Bargmann--Fock space}. II, J. Reine Angew.
Math \textbf{429} (1992), 107--113.

\bibitem {wi36}J. Williamson, \textit{On the algebraic problem concerning the
normal forms of linear dynamical systems}, Amer. J. of Math. 58 (1936)
141--163.%
\[
\]

\end{thebibliography}
\end{document}